\documentclass[reqno,centertags, 12pt]{amsart}
\usepackage{amsmath,amsthm,amscd,amssymb}
\usepackage{latexsym}
\newcommand{\bbC}{{\mathbb{C}}}
\newcommand{\bbD}{{\mathbb{D}}}

\newcommand{\bbR}{{\mathbb{R}}}

\newcommand{\fre}{{\frak{e}}}

\newcommand{\calI}{{\mathcal I}}

\newcommand{\no}{\nonumber}
\newcommand{\lb}{\label}
\newcommand{\f}{\frac}

\newcommand{\ti}{\tilde  }

\newcommand{\dist}{\text{\rm{dist}}}
\newcommand{\intt}{\text{\rm{int}}}

\newcommand{\s}{\text{\rm{s}}}

\newcommand{\pp}{\text{\rm{pp}}}
\newcommand{\supp}{\text{\rm{supp}}}

\newcommand{\bi}{\bibitem}

\newcommand{\beq}{\begin{equation}}
\newcommand{\eeq}{\end{equation}}
\newcommand{\ba}{\begin{align}}
\newcommand{\ea}{\end{align}}
\newcommand{\veps}{\varepsilon}

\newcounter{smalllist}
\newenvironment{SL}{\begin{list}{{\rm\roman{smalllist})}}{%
\setlength{\topsep}{0mm}\setlength{\parsep}{0mm}\setlength{\itemsep}{0mm}%
\setlength{\labelwidth}{2em}\setlength{\leftmargin}{2em}\usecounter{smalllist}%
}}{\end{list}}

\DeclareMathOperator{\Real}{Re}
\DeclareMathOperator{\Ima}{Im}
\DeclareMathOperator{\diam}{diam}

\allowdisplaybreaks
\numberwithin{equation}{section}

\newtheorem{theorem}{Theorem}[section]

\newtheorem*{p2.1}{Proposition 2.1}
\newtheorem{proposition}[theorem]{Proposition}
\newtheorem{lemma}[theorem]{Lemma}
\newtheorem{corollary}[theorem]{Corollary}
\theoremstyle{definition}

\theoremstyle{remark}
\newtheorem*{remark}{Remark}
\newtheorem*{remarks}{Remarks}
\newtheorem*{definition}{Definition}

\newcommand{\abs}[1]{\lvert#1\rvert}

\newcommand{\norm}[1]{\lVert#1\rVert}

\begin{document}
\title{The Hilbert Transform of a Measure}
\author[A.~Poltoratski, B.~Simon, and M.~Zinchenko]
{Alexei Poltoratski$^{1,2}$, Barry Simon$^{3,4}$, and
Maxim Zinchenko$^3$}

\thanks{$^1$ Mathematics Department, Texas A\&M University, College Station, TX 77843, USA. E-mail:
alexei@math.tamu.edu}

\thanks{$^2$ Supported in part by NSF grant DMS-0800300}

\thanks{$^3$ Mathematics 253-37, California Institute of Technology, Pasadena, CA 91125, USA.
E-mail: bsimon@caltech.edu; maxim@caltech.edu}

\thanks{$^4$ Supported in part by NSF grant DMS-0652919}

\date{May 29, 2009}
\keywords{Hilbert transform, homogeneous set, weak $L^1$}
\subjclass[2000]{42A50, 26A30, 42B25}

\begin{abstract}
Let $\fre$ be a homogeneous subset of $\bbR$ in the sense of Carleson. Let
$\mu$ be a finite positive measure on $\bbR$ and $H_\mu(x)$ its Hilbert
transform. We prove that if $\lim_{t\to\infty} t \abs{\fre\cap\{x\mid
\abs{H_\mu(x)}
>t\}}=0$, then $\mu_s(\fre)=0$, where $\mu_\s$ is the singular part
of $\mu$.
\end{abstract}

\maketitle

\section{Introduction} \lb{s1}

This is a paper about the Hilbert transform of a measure defined as follows. The Stieltjes transform
(also called Borel transform or Markov function) of a finite (positive) measure, $\mu$, is defined on
$\bbC_+ =\{z\mid\Ima z>0\}$ by
\begin{equation} \lb{1.1}
F_\mu(z)=\int \f{d\mu(x)}{x-z}
\end{equation}
For Lebesgue a.e.\ $x\in\bbR$,
\begin{equation} \lb{1.2}
F_\mu(x+i0) =\lim_{\veps\downarrow 0}\, F_\mu(x+i\veps)
\end{equation}
exists. The Hilbert transform is given by
\begin{equation} \lb{1.3}
H_\mu(x)=\f{1}{\pi}\, \Real F_\mu(x+i0)
\end{equation}

It is a result of Loomis \cite{Loo} that for a universal constant, $C$, ($\norm{\mu}\equiv\mu(\bbR)$)
\begin{equation} \lb{1.4}
\abs{\{x\mid \abs{H_\mu(x)}\geq t\}} \leq \f{C \norm{\mu}}{t}
\end{equation}
This was earlier proven for the a.c.\ case by Kolmogorov (attributed by Zygmund
\cite{Zyg}) and, for finite point measures, Boole \cite{Boole} proved (and
Loomis rediscovered)
\begin{equation} \lb{1.5}
\abs{\{x\mid \pm H_\mu(x) \geq t\}} = \f{\norm{\mu}}{\pi t}
\end{equation}
We note that \eqref{1.5} was extended by Hru{\v{s}}{\v{c}}{\"e}v--Vinogradev \cite{HV} to all singular
measures; see also \cite{S250,Pol}.

\begin{remark}
We do not need an explicit value of $C$ in \eqref{1.4}. Davis
\cite{Davis,Da74} has shown the optimal constant in \eqref{1.4} is
$C=1$.
\end{remark}

In distinction, for a.c. measures, $d\mu=f\,dx$, we have
\begin{equation} \lb{1.6}
\lim_{t\to\infty}\, t \abs{\{x\mid \abs{H_{f\,dx}(x)}\geq t\}} =0
\end{equation}
This follows from the fact that if $f\in L^2$, $H_{f\,dx}\in L^2$ (indeed,
$\norm{H_{f\,dx}}_2 = \norm{f}_2$), that $L^2\cap L^1$ is dense in $L^1$, that
\eqref{1.6} is trivial if $H_{f\,dx}$ is $L^2$, and that for any $\theta\in
[0,1]$,
\begin{equation} \lb{1.7}
\begin{split}
\abs{\{x\mid  \abs{f(x) & +g(x)} > t\}} \\
&\leq \abs{\{x\mid \abs{f(x)} >\theta t\}} + \abs{\{x\mid \abs{g(x)} > (1-\theta)t\}}
\end{split}
\end{equation}
From \eqref{1.5}, \eqref{1.6}, and \eqref{1.7}, one sees
\begin{equation} \lb{1.8}
\lim_{t\to +\infty}\, \pi t \abs{\{x\mid \pm H_\mu(x)\geq t\}} = \norm{\mu_\s}
\end{equation}
where
\begin{equation} \lb{1.9}
d\mu = f\, dx + d\mu_\s
\end{equation}
is the Lebesgue decomposition of $\mu$ (i.e., $\mu_\s$ is singular). We also note that \eqref{1.8} is a special case of Poltoratski's formula \eqref{5.4}.

One can rephrase this. We recall that weak-$L^1$ is defined by (this is not a norm!) setting
\begin{equation} \lb{1.10}
\norm{f}_{1,w}\equiv \sup_t\, t \abs{\{x\mid \abs{f(x)}\geq t\}}
\end{equation}
and
\begin{equation} \lb{1.11}
L_w^1 = \{f\mid\norm{f}_{1,w}<\infty\}
\end{equation}
so \eqref{1.4} says $H_\mu\in L_w^1$. We also define
\begin{equation} \lb{1.12}
L_{w;0}^1 = \bigl\{f\in L_w^1\mid \lim_{t\to\infty}\, t \abs{\{x\mid\abs{f(x)}\geq t\}}=0\bigr\}
\end{equation}
and \eqref{1.8} implies
\begin{equation} \lb{1.13}
H_\mu\in L_{w;0}^1 \Leftrightarrow \mu_s(\bbR)=0
\end{equation}

Our main goal is to provide a local version of this theorem for special sets singled out by Carleson \cite{Car}.

\begin{definition}
We say that a compact set $\fre\subset\bbR$ is {\it homogeneous} (with
homogeneity constant $\delta$) if there is $\delta >0$, such that
for all $x\in\fre$ and $0<a<\diam(\fre)$,
\begin{equation} \lb{1.14}
\abs{\fre\cap(x-a,x+a)} \geq 2\delta a
\end{equation}
\end{definition}

Given a function, $f$, we use $f\restriction\fre$ to denote the function $f\chi_\fre$ with $\chi_\fre$ the
characteristic function of $\fre$. The purpose of this paper is to prove

\begin{theorem}\lb{T1.1} Let $\fre$ be homogeneous and let $\mu$ be a measure on $\bbR$ so that $H_\mu\restriction
\fre\in L_{w;0}^1$. Then
\begin{equation} \lb{1.15}
\mu_\s(\fre)=0
\end{equation}
\end{theorem}

\begin{remarks} 1. There is an analog for measures on $\partial\bbD=\{z\in\bbC\mid \abs{z}=1\}$.

\smallskip
2. The Hilbert transform can be defined if $\mu$, rather than being finite, obeys $\int (1+\abs{x})^{-1}\,
d\mu <\infty$. Indeed, $H_\mu$ can be defined up to an additive constant if $\int (1+\abs{x}^2)^{-1}\,
d\mu(x)<\infty$. Theorem~\ref{T1.1} extends to both these cases.

\smallskip
3. It follows from the arguments in Section~\ref{s2} that a converse to Theorem~\ref{T1.1} holds and
that $H_\mu\restriction\fre\in L_{w;0}^1$ if and only if $H_{\mu\restriction\fre}\in L_{w;0}^1$. Thus,
we have a three-fold equivalence,
\begin{equation} \lb{1.15a}
H_\mu\restriction\fre\in L_{w;0}^1 \Leftrightarrow H_{\mu\restriction\fre}\in L_{w;0}^1
\Leftrightarrow \mu_\s(\fre) =0
\end{equation}
\end{remarks}

There is a special case that is both important and one motivation for this work. We recall \cite{NVY}:

\begin{definition} A finite measure $\mu$ on $\bbR$ is called {\it reflectionless\/} on $\fre\subset\bbR$,
where $\fre$ is compact and of strictly positive Lebesgue measure, if and only if $H_\mu\restriction\fre=0$.
\end{definition}

There has been an explosion of recent interest about reflectionless measures due to work of
Remling \cite{Remppt}. Clearly, the zero function lies in $L_{w;0}^1$, so

\begin{corollary}\lb{C1.2} Let $\fre$ be homogeneous; let $\mu$ be a measure on $\bbR$ which is reflectionless
on $\fre$. Then \eqref{1.15} holds.
\end{corollary}

This result is not new. For cases where $\supp(\mu)\subset\fre$, it is due to Sodin--Yuditskii \cite{SY}, with
some extensions due to Gesztesy--Zinchenko \cite{GZ}. Recently, Poltoratski--Remling \cite{PR} have proven a
stronger result than Corollary~\ref{C1.2}---instead of requiring that $\fre$ is homogeneous, they only need for
all $x_0\in\fre$ that
\begin{equation} \lb{1.16}
\limsup_{a\downarrow 0}\, \f{\abs{\fre\cap(x_0-a,x_0+a)}}{2a} >0
\end{equation}
If \eqref{1.16} holds for all $x_0\in\fre$, we call $\fre$ {\it weakly homogeneous}, following \cite{PR}.

The property of being reflectionless is not robust in that changing $\mu$ off $\fre$ will usually destroy the
reflectionless property. As we will see in Section~\ref{s2}, having $H_\mu\restriction\fre$ in $L_{w;0}^1$ is
robust and explains one reason we sought this result.

Our proof is quite different from \cite{PR}. We note, however, that
our proof, like the one in \cite{PR}, is essentially a real variable
proof (we go into the complex plane but use no contour integrals),
while the earlier work of \cite{SY,GZ} is a complex variable
argument.

We mention that Corollary~\ref{C1.2} (and so Theorem~\ref{T1.1}) does not hold
for arbitrary $\fre$. Nazarov--Volberg--Yuditkii \cite{NVY} have examples of
reflectionless measures on their supports where \eqref{1.16} fails and that
have a singular component.

We want to mention another special case of Theorem~\ref{T1.1}:

\begin{corollary} \lb{C1.3}
Let $\fre$ be a homogeneous set in $\bbR$. Let $\mu$ be a measure on $\bbR$ so
that there is a set $A$ with
\begin{SL}
\item[{\rm{(i)}}] $\abs{A}=0$
\item[{\rm{(ii)}}] $\mu(\bbR\setminus A)=0$
\item[{\rm{(iii)}}] $A$ is closed and $A\subset\fre$
\end{SL}
Suppose $H_\mu\restriction\fre\in L_{w;0}^1$. Then $\mu=0$.
\end{corollary}

We will need a strengthening of this special case:

\begin{theorem}\lb{T1.4}
Let $\fre$ be a homogeneous set in $\bbR$. There is a constant $C_1$
depending only on $\fre$ so that for any measure, $\mu$, obeying
{\rm{(i)--(iii)}} of Corollary~\ref{C1.3}, we have that
\begin{equation} \lb{1.17}
\mu(\fre) \leq C_1 \, \liminf_{t\to\infty}\, t\abs{\{x\in\fre\mid\abs{H_\mu(x)}\geq t\}}
\end{equation}
\end{theorem}

\begin{remarks} 1. In fact, $C_1$ is only $\delta$-dependent; explicitly, one can take
\begin{equation} \lb{1.18}
C_1 = \f{1536\pi^3}{\delta^2}
\end{equation}
We have made no attempt to optimize this constant and, indeed, have made choices to simplify the
arithmetic. The $\delta^{-2}$ may be optimal, and certainly it seems that $\delta^{-1}$ is not
possible.

\smallskip
2. There is also a strengthening of Theorem~\ref{T1.1} of this same form.
\end{remarks}

We can say more about weakly homogeneous sets, that is, ones that obey \eqref{1.16}, and thereby
illuminate and limit Theorem~\ref{T1.1}.

\begin{theorem} \lb{T1.5} Let $\fre$ be a compact weakly homogeneous set and $\mu$ a measure
on $\bbR$ so that $H_\mu\restriction\fre\in L_{w;0}^1$. Then for all $x_0\in\fre$,
\begin{equation}\lb{1.19}
\mu(\{x_0\}) =0
\end{equation}
that is, $\mu$ has no pure points in $\fre$.
\end{theorem}

\begin{theorem}\lb{T1.6} There exists a weakly homogeneous set, $\fre$, containing the classical
Cantor set so that if $\mu$ is the conventional Cantor measure, $H_\mu\restriction\fre\in L_{w;0}^1$.

In particular, Theorem~\ref{T1.1} does not extend to weakly homogeneous sets.
\end{theorem}

While the gap between homogeneous and weakly homogeneous sets is not large, we can extend Theorem~\ref{T1.1}
to partly fill it in. We call a set, $\fre$, {\it non-uniformly homogeneous\/} if it is closed and obeys
\begin{equation} \lb{1.21}
\liminf_{a\downarrow 0}\, (2a)^{-1} \abs{\fre\cap (x-a,x+a)} >0
\end{equation}
for all $x\in\fre$.

\begin{theorem} \lb{T1.7} Let $\fre$ be non-uniformly homogeneous and let $\mu$ be a measure on $\bbR$ so
that $H_\mu\restriction\fre\in L_{w;0}^1$. Then
\begin{equation} \lb{1.22}
\mu_\s(\fre)=0
\end{equation}
\end{theorem}

In fact, we will obtain this from a stronger result. We emphasize that $\fre$ in the next theorem is not assumed
closed.

\begin{theorem} \lb{T1.8} Let $\fre$ be a Borel set in $\bbR$ and $\mu$ a finite measure so that $H_\mu\restriction
\fre\in L_{w;0}^1$. Then
\begin{equation} \lb{1.22x}
\mu_\s \bigl( \{x\in\fre\bigm| \liminf_{a\downarrow 0}\, (2a)^{-1} \abs{\fre\cap (x-a,x+a)} >0\}\bigr) =0
\end{equation}
\end{theorem}

This is to be compared with the result of Poltoratski--Remling \cite{PR} that if $\fre$ is Borel and $H_\mu
\restriction\fre=0$, then
\begin{equation} \lb{1.23}
\mu_\s \bigl( \{x\in\fre\bigm| \limsup_{a\downarrow 0}\, (2a)^{-1} \abs{\fre\cap (x-a,x+a)} >0\}\bigr) =0
\end{equation}
and the statement that follows from our proof of Theorem~\ref{T1.5} that if $\mu_\pp$ is the pure point part
of $\mu$, then if $H_\mu\restriction\fre\in L_{w;0}^1$, then
\[
\mu_\pp \bigl( \{x\in\fre\bigm| \limsup_{a\downarrow 0}\, (2a)^{-1} \abs{\fre\cap (x-a,x+a)} >0\}\bigr) =0
\]
Moreover, it is to be noted that the example in Theorem~\ref{T1.6} shows that in Theorem~\ref{T1.8}, we cannot
replace \eqref{1.22x} by \eqref{1.23}.

In Section~\ref{s2}, we reduce the proof of Theorem~\ref{T1.1} to proving Theorem~\ref{T1.4}. In
Section~\ref{s3}, we prove Theorem~\ref{T1.4}. In proving Theorem~\ref{T1.4}, we first show that if
$[a,b]$ is an interval on which $\abs{F_\mu(x+i0)}\geq t$, then $\abs{F_\mu (x+i(b-a))}\geq t/8\pi^2$.
Then we will use this to prove that on most of $[a-(b-a),a]$ and $[b,b+(b-a)]$, $\abs{F_\mu(x+i0)}$
is a significant fraction of $t$, which is the key to the proof. In Section~\ref{s4}, we prove
Theorems~\ref{T1.5} and \ref{T1.6}. In Section~\ref{s5}, we prove Theorem~\ref{T1.8}, and so
Theorem~\ref{T1.7}.

\medskip
We want to thank Jonathan Breuer and Yoram Last for useful discussions.

\section{Reduction to Theorem~\ref{T1.4} } \lb{s2}

In this section, we show that Theorem~\ref{T1.4} implies Theorem~\ref{T1.1}.

\begin{proposition} \lb{P2.1} Let $\mu$ have the form \eqref{1.9}. Then for any set $\fre\subset\bbR$,
\begin{equation} \lb{2.1}
H_\mu\restriction\fre\in L_{w;0}^1\Leftrightarrow H_{\mu_\s}\restriction\fre \in L_{w;0}^1
\end{equation}
In particular, we need only prove Theorem~\ref{T1.1} for purely singular measures to get it for all measures.
\end{proposition}

\begin{remark}  This shows the advantage of working with $L_{w;0}^1$. Purely singular measures are never
reflectionless (for $\abs{\{x\mid F_\mu(x+i0)=0\}}=0$ and thus, $\Ima F_\mu (x+i0)>0$ a.e. on $\fre$ if
$H_\mu\restriction\fre =0$).
\end{remark}

\begin{proof} By \eqref{1.7} with $\theta=\f12$, $L_{w;0}^1$ is a vector space. Since $H_\mu-H_{\mu_\s}
= H_{f\,d\mu}\in L_{w;0}^1$, by \eqref{1.6}, we get \eqref{2.1} immediately.
\end{proof}

\begin{proposition} \lb{P2.2} Let $\fre$ be a closed set. Let $\mu$ be a measure with $\mu(\fre)=0$. Then
\begin{equation} \lb{2.2}
H_\mu\restriction\fre\in L_{w;0}^1
\end{equation}
\end{proposition}

\begin{proof} Let $\mu_m =\mu\restriction\{x\mid\dist(x,\fre)\geq m^{-1}\}$. Then for $x\in\fre$,
\begin{equation} \lb{2.3}
H_{\mu_m}(x) = \f{1}{\pi} \int \f{d\mu_m(y)}{y-x}
\end{equation}
so
\begin{equation} \lb{2.4}
\norm{H_{\mu_m}\restriction\fre}_\infty \leq \f{m}{\pi}\, \norm{\mu_m}
\end{equation}
so $H_{\mu_m}\in L_{w;0}^1$.

By \eqref{1.7} with $\theta=\f12$, for any $m$,
\begin{align}
\limsup_{t\to\infty}\, t \abs{\{x\in\fre\mid\abs{H_\mu(x)}\geq t\}}
&\leq 2\limsup_{t\to\infty}\, t\abs{\{x\in\fre\mid\abs{H_{\mu-\mu_m}(x)}\geq t\}}
\no \\
&\leq 2 C\norm{\mu-\mu_m} \lb{2.6}
\end{align}
where $C$ is the constant in \eqref{1.4}.

Since \eqref{2.6} holds for all $m$ and $\norm{\mu-\mu_m}\to 0$ (since $\mu(\fre)=0$), we conclude $H_\mu
\restriction\fre\in L_{w;0}^1$.
\end{proof}

\begin{proposition} \lb{P2.3} Let $\fre$ be a closed set. Let $\nu=\mu\restriction\fre$, that is, $\nu(A)
=\mu(\fre\cap A)$. Then
\begin{equation} \lb{2.7}
H_\nu\restriction\fre \in L_{w;0}^1 \Leftrightarrow H_\mu\restriction\fre\in L_{w;0}^1
\end{equation}
In particular, it suffices to prove Theorem~\ref{T1.1} for purely singular measures supported on $\fre$.
\end{proposition}

\begin{proof} Let $\eta=\mu-\nu$. By Proposition~\ref{P2.2},
\begin{equation} \lb{2.8}
H_\mu\restriction\fre - H_\nu \restriction\fre = H_\eta\restriction\fre\in L_{w;0}^1
\end{equation}
Since $L_{w;0}^1$ is a vector space, \eqref{2.8} implies \eqref{2.7}.
\end{proof}

\begin{proof}[Proof of Theorem~\ref{T1.1} given Theorem~\ref{T1.4}] By Proposition~\ref{2.3}, we can suppose
$\mu$ is purely singular and supported by $\fre$. Thus, there exists $A_\infty\subset\fre$ with $\abs{A_\infty}
=0$, so $\mu(\bbR\setminus A_\infty)=0$.

By regularity of measures, we can find $A_n\subset A_{n+1}\subset\cdots\subset A_\infty$ with each $A_n$
closed, and so
\begin{equation} \lb{2.9}
\mu(A_\infty\setminus A_n)\to 0
\end{equation}

Define $\mu_n =\mu\restriction A_n$ and $\nu_n=\mu-\mu_n$. By \eqref{1.7} with $\theta=\f12$, $H_\mu\restriction
\fre\in L_{w;0}^1$, and \eqref{1.4},
\begin{align}
\limsup_{t\to\infty}\, t \abs{\{x\in\fre\mid\abs{H_{\mu_n}(x)}\geq t\}}
&\leq 2\limsup_{t\to\infty}\, t \abs{\{x\in\fre\mid\abs{H_{\nu_n}(x)}\geq t\}} \notag \\
&\leq 2C\mu(A_\infty\setminus A_n) \lb{2.10}
\end{align}

$A_n$ obeys (i)--(iii) for $\mu_n$, so by \eqref{1.17},
\begin{equation} \lb{2.11}
\mu(A_n) =\mu_n(\fre) \leq 2CC_1 \mu(A_\infty\setminus A_n)
\end{equation}

As $n\to\infty$, $\mu(A_n) \to \mu_\s(\fre)$ while, by \eqref{2.9}, $\mu(A_\infty\setminus A_n) \to 0$.
So $\mu_\s(\fre)=0$.
\end{proof}

\section{Proof of Theorem~\ref{T1.4}} \lb{s3}

Throughout this section, where we will prove Theorem~\ref{T1.4} and so complete the proof of Theorem~\ref{T1.1},
we suppose $\fre$ is homogeneous with homogeneity constant $\delta$, and $\mu$ is a measure for which there
exists $A\subset\fre$ obeying properties (i)--(iii) of Corollary~\ref{C1.3}. In particular, since $\mu$ is
 singular, for a.e.\ $x\in\bbR$,
\begin{equation} \lb{3.1}
F_\mu (x+i0)=\pi H_\mu(x)
\end{equation}
We will consider $F_\mu$ throughout.

The key will be to prove for all large $t$,
\begin{equation} \lb{3.2}
\biggl|\biggl\{x\in\fre\biggm| \abs{F_\mu(x+i0)} > \f{\delta}{128\pi^2} t\biggr\}\biggr| \geq
\f{\delta}{24}\, \abs{\{x\mid\abs{F_\mu(x+i0)} > t\}}
\end{equation}
We will do this by showing that if $I$ is an interval in
$\bbR\setminus A$ where $\abs{F_\mu (x+i0)}>t$, then at most points
of the two touching intervals of the same size, $\abs{F_\mu}\geq
\delta t/128\pi^2$. We will do this in two steps. We show that $F(z)$ at
points over $I$ with $\Ima z=\abs{I}$ is comparable to $t$ and use
that to control $F$ on the touching intervals. A Vitali covering map
argument will boost that up to the full sets. We need

\begin{proposition}\lb{P3.2} Let
\begin{equation} \lb{3.7}
I=[c-a,c+a]
\end{equation}
be an interval contained in
\begin{equation} \lb{3.8}
\{x\mid\abs{F_\mu(x+i0)}\geq t\}
\end{equation}
Then
\begin{equation} \lb{3.9}
\abs{F_\mu(c+a+2ia)} \geq \f{t}{8\pi^2}
\end{equation}
\end{proposition}

\begin{proof} $F_\mu$ lies in weak $L^1$ and is bounded off a compact subset of $\bbR$. For $z\in\bbC_+$, let
\begin{equation} \lb{3.10}
G(z) =\sqrt{F_\mu(z)/i}
\end{equation}
Then $G$ has locally $L^1$ boundary values on $\bbR$ and is bounded off a compact set, so if $z=x+iy$,
\begin{equation} \lb{3.11}
G(z)=\f{1}{\pi}\int \f{yG(\lambda +i0)\,d\lambda}{(x-\lambda)^2 +y^2}
\end{equation}

$\arg (G)\in [-\f{\pi}{4}, \f{\pi}{4}]$, so on $\bbR$,
\begin{equation} \lb{3.12}
\Real G(\lambda+i0)\geq 0
\end{equation}
On $I$, $\arg (G)=\pm \f{\pi}{4}$, and so for $\lambda\in I$,
\begin{equation} \lb{3.13}
\Real G(\lambda+i0) \geq \sqrt{t/2}
\end{equation}

Thus, by \eqref{3.11}, \eqref{3.12}, and, \eqref{3.13},
\begin{align}
\Real G(c+a+2ia) &\geq \f{1}{\pi} \int_I \f{2a\Real G(\lambda+i0)}{(c+a-\lambda)^2 + (2a)^2}\, d\lambda \notag \\
&\geq \f{1}{\pi} \f{(2a)^2 \sqrt{t/2}}{(2a)^2+(2a)^2} \geq \f{1}{2\pi} \, \sqrt{t/2} \lb{3.14}
\end{align}
so
\begin{equation} \lb{3.15}
\abs{F_\mu(c+a+2ia)} \geq (\Real G(c+a+2ia))^2 \geq \f{t}{8\pi^2}
\end{equation}
\end{proof}

\begin{lemma} \lb{L3.3} Fix $t_0 >0$ and let
\begin{equation} \lb{3.16}
F_{t_0}(z) = \f{F(z)}{1+\f{1}{t_0} F(z)}
\end{equation}
Then, $\Ima F_{t_0}>0$ on $\bbC_+$ and
\begin{equation} \lb{3.17x}
\{x\mid \abs{F(x+i0)} >t_0\} = \biggl\{x\biggm| F_{t_0}(x+i0) > \f{t_0}{2}\biggr\}
\end{equation}
\end{lemma}

\begin{remark} $F_{t_0}$ is the Stieltjes transform of a measure associated with a rank one perturbation
(see, e.g., \cite[Sect.~11.2]{TI}), but that will play no direct role here.
\end{remark}

\begin{proof} The invertible map
\begin{equation} \lb{3.17}
H(z)=\f{z}{1+\f{z}{t_0}}
\end{equation}
maps $\bbC_+$ to $\bbC_+$ and $(t_0,\infty)\cup\{\infty\}\cup(-\infty,-t_0)$ to $(\f{t_0}{2},\infty)$.
\end{proof}

For any $x>0$, define
\begin{equation} \lb{3.18}
\Gamma_s =\{x\mid\abs{F(x+i0)} >s\}
\end{equation}

\begin{proposition}\lb{P3.4} Fix $t>0$ and let
\begin{equation} \lb{3.19}
t_0=\f{\delta}{128\pi^2}\, t
\end{equation}
Suppose
\begin{equation} \lb{3.20}
I=[c-a,c+a]\subseteq\Gamma_t
\end{equation}
and let
\begin{equation} \lb{3.21}
\ti I =[c+a,c+3a]
\end{equation}
be the touching interval of the same size as $I$. Then
\begin{equation} \lb{3.22}
\abs{\ti I\setminus\Gamma_{t_0}} \leq a\delta = \f{\delta}{2}\, \abs{I}
\end{equation}
\end{proposition}

\begin{proof} By the lemma for $x$ real,
\begin{equation} \lb{3.23}
\chi_{\Gamma_{t_0}}(x) = 1-\f{1}{\pi} \, \arg\biggl( F_{t_0} (x+i0)-\f{t_0}{2}\biggr)
\end{equation}
which is the boundary value of a bounded harmonic function.

Let
\begin{equation} \lb{3.24}
z_0 =c+a+2ia
\end{equation}
Then
\begin{align}
\arg\biggl(F_{t_0}(z_0)-\f{t_0}{2}\biggr)
&= \arg\biggl( \f{F(z_0) - \f{t_0}{2} - \f{F(z_0)}{2}}{1+\f{1}{t_0}F(z_0)}\biggr) \notag \\
&= \arg\biggl( \f{\f{F(z_0)}{t_0}-1}{\f{F(z_0)}{t_0}+1}\biggr) =
\arg\biggl( 1-\f{2}{\f{F(z_0)}{t_0}+1}\biggr) \lb{3.25}
\end{align}

By Proposition~\ref{P3.2},
\begin{equation} \lb{3.26}
\biggl| \f{F(z_0)}{t_0}\biggr| \geq \f{t}{8\pi^2t_0} = \f{16}{\delta} \geq 16
\end{equation}
since $\delta \leq 1$. Thus,
\begin{equation} \lb{3.27}
\biggl| \f{2}{\f{F(z_0)}{t_0}+1}\biggr| \leq \f{2}{\abs{\f{F(z_0)}{t_0}}-1} \leq \f{2}{15} < 1
\end{equation}

If $\abs{w}\leq 1$ for $w\in\bbC$, then
\begin{equation} \lb{3.28}
\arg(1+w) \leq \arcsin (\abs{w}) \leq \f{\pi}{2}\, \abs{w}
\end{equation}
($\sin(y) \geq \f{2y}{\pi}$ for $y\in [0,\f{\pi}{2}]$ implies for
$x\in [0,1]$, $\arcsin x \leq \f{\pi}{2}x$). By \eqref{3.25},
\begin{equation} \lb{3.29}
\arg \biggl(F_{t_0}(z_0) - \f{t_0}{2}\biggr) \leq \f{8\pi^3 t_0}{t-8\pi^2 t_0}
\end{equation}

Thus, if $\chi_{\Gamma_{t_0}}(z)$ is the harmonic function whose boundary value is $\chi_{\Gamma_{t_0}}(x)$, we find, by
\eqref{3.23}, that
\begin{equation} \lb{3.30}
\pi (1-\chi_{\Gamma_{t_0}}(z_0)) \leq \f{8\pi^3 t_0}{t-8\pi^2 t_0}
\end{equation}

By a Poisson formula with $z_0=x_0+iy_0$ as in \eqref{3.24},
\begin{align}
\pi(1-\chi_{\Gamma_{t_0}}(z_0)) &= \int_{\bbR\setminus\Gamma_{t_0}}
\f{y_0\, d\lambda}{(\lambda-x_0)^2 + y_0^2} \lb{3.31} \\
&\geq \f{1}{2}\, \f{\abs{\ti I\setminus\Gamma_{t_0}}}{\abs{I}} \lb{3.32}
\end{align}
since on $\ti I$, the minimum of $y_0/((\lambda-x_0)^2 + y_0^2)$ is
$1/(2\abs{I})$.

Thus, by \eqref{3.30} and \eqref{3.32},
\begin{equation} \lb{3.33}
\abs{\ti I\setminus \Gamma_{t_0}} \leq \f{16\pi^3 t_0}{t-8\pi^2 t_0}\, \abs{I}
\end{equation}

Since $\f{8\pi^2 t_0}{t}\leq \f{1}{16}$,
\[
\f{\f{16\pi^3 t_0}{t}}{1-\f{8\pi^2 t_0}{t}} \leq \f{256\pi^3 t_0}{15\, t} = \f{4\pi}{15} \, \f{\delta}{2} \leq \f{\delta}{2}
\]
and \eqref{3.33} implies \eqref{3.22}.
\end{proof}

\begin{proposition} \lb{P3.5} Under the notation of Proposition~\ref{P3.4}, let
\begin{equation} \lb{3.34}
I^\sharp = [c-3a, c+3a]
\end{equation}
and suppose
\begin{equation} \lb{3.35}
\fre\cap I\neq \emptyset
\end{equation}
and
\begin{equation} \lb{3.36}
a\leq \diam(\fre)
\end{equation}
Then
\begin{equation} \lb{3.37}
\abs{\Gamma_{t_0} \cap\fre\cap I^\sharp} \geq \f{\delta}{2}\, \abs{I}
\end{equation}
\end{proposition}

\begin{proof} Pick $x_0\in\fre\cap I$. Suppose $x_0\geq c$. If not, we pick $\ti I$ to be the third of
$I^\sharp$ below $I$ instead of the choice here. By homogeneity,
\begin{equation} \lb{3.38}
\abs{\fre\cap (x_0-a,x_0+a)} \geq 2a\delta =\delta\abs{I}
\end{equation}
and the intersection lies in $I\cup\ti I$. Thus,
\begin{equation} \lb{3.39}
\abs{\Gamma_{t_0}\cap\fre\cap I^\sharp} \geq \abs{\fre\cap (x_0-a,x_0+a)} -
\abs{(I\cup\ti I)\setminus\Gamma_{t_0}}
\end{equation}
Since $I\subset\Gamma_t\subset\Gamma_{t_0}$,
\begin{equation} \lb{3.40}
\abs{(I\cup\ti I)\setminus\Gamma_{t_0}} = \abs{\ti I\setminus\Gamma_{t_0}} \leq \f{\delta}{2}\, \abs{I}
\end{equation}
by \eqref{3.22}. \eqref{3.38} and \eqref{3.39} imply \eqref{3.37}.
\end{proof}

\begin{proof}[Proof of Theorem~\ref{T1.4}] Suppose $\mu\neq 0$. On $\bbR\setminus A$, $F_\mu(x+i0)$ is
continuous and real, so $\{x\mid \abs{F_\mu(x+i0)} >t\}$ is open, and so a countable union of maximal
disjoint open intervals.

Let $I=[c-a,c+a]$ be the closure of any such interval. On $\bbR\setminus A$, $F_\mu(x)$ has
\begin{equation} \lb{3.41}
F'_\mu(x) = \int \f{d\mu(x)}{(y-x)^2} >0
\end{equation}
If $F_\mu >t$ on $I$, $c+a$ must be in $A$ or else $F_\mu(c+a)<\infty$ and $F_\mu(c+a+\veps)\in\Gamma_t$ for
$\veps$ small (so $I$ is not maximal). Similarly, if $F_\mu <-t$ on $I$, $c-a\in A$. Thus, $I\cap A\neq
\emptyset$, so $I\cap\fre\neq\emptyset$.

Let
\begin{equation} \lb{3.42}
T=\f{\pi C\norm{\mu}}{\diam(\fre)}
\end{equation}
where $C$ is the constant in \eqref{1.4}. Then for $t>T$, $\abs{\Gamma_t}\leq \diam(\fre)$, so $a\leq
\diam(\fre)$. Thus, by Proposition~\ref{P3.5},
\begin{equation} \lb{3.43}
\abs{\Gamma_{t_0}\cap\fre\cap I^\sharp} \geq \f{\delta}{2}\, \abs{I}
\end{equation}

Clearly, the $I$'s and so the $(I^\sharp)^\intt$'s are an open cover of $\Gamma_t\setminus A$. Thus, by
the Vitali covering theorem (see Rudin \cite[Lemma~7.3]{Rudin}), we can find a subset of mutually disjoint
$I^\sharp$'s, call them $\{I_j^\sharp\}$, so that
\begin{equation} \lb{3.44}
\abs{\Gamma_t} \leq 4 \sum_j\, \abs{I_j^\sharp} \leq 12\sum_j\, \abs{I_j}
\end{equation}

By the disjointness, with $t_0$ given by \eqref{3.19},
\begin{alignat*}{2}
\abs{\Gamma_{t_0}\cap\fre} &\geq \sum_j\, \abs{I_j^\sharp\cap\Gamma_{t_0}\cap\fre} \notag \\
&\geq \f{\delta}{2}\sum_j\, \abs{I_j} \qquad && \text{(by \eqref{3.37})} \\
&\geq \f{\delta}{24}\, \abs{\Gamma_t} \qquad && \text{(by \eqref{3.44})}
\end{alignat*}
Thus,
\[
\liminf_{t\to\infty}\, t_0 \abs{\Gamma_{t_0}\cap\fre} \geq \liminf_{t\to\infty}\, \f{\delta}{24}\,
\f{\delta}{128\pi^2}\, t\abs{\Gamma_t}
\]
Therefore, by \eqref{1.8} and \eqref{3.1},
\[
\liminf_{t\to\infty}\, t \abs{\{x\in\fre\mid\abs{H_\mu(x)} >t\}} \geq \f{\delta^2}{3072\pi^2}\,
\f{2(\mu(A))}{\pi}
\]
which is \eqref{1.17}/\eqref{1.18}.
\end{proof}

\section{Weakly Homogeneous Sets} \lb{s4}

\begin{proof}[Proof of Theorem~\ref{T1.5}] For $x_0\in\fre$ and $\veps >0$, write
\begin{equation} \lb{4.1}
\mu = \mu_1 + \mu_2 + \mu_3
\end{equation}
with $\mu_1 =\mu\restriction\{x_0\}$, $\mu_2 = \mu\restriction [(x_0-\veps, x_0 + \veps)\setminus
\{x_0\}]$, $\mu_3 =\mu\restriction \bbR\setminus (x_0-\veps, x_0+\veps)$, and by \eqref{1.7}, note
\begin{equation} \lb{4.2}
\begin{split}
\abs{\{x\in\fre;\, &\abs{x-x_0} < \tfrac{\veps}{2}\mid \abs{H_{\mu_1}(x)} > 3t\}}
 \leq  \abs{\{x\in\fre\mid \abs{H_\mu(x)} > t\}} \\
& + \abs{\{x \mid \abs{H_{\mu_2}(x)}>t\}}
+ \abs{\{x;\, \abs{x-x_0} < \tfrac{\veps}{2}\mid \abs{H_{\mu_3}(x)} >t\}}
\end{split}
\end{equation}
By hypothesis, the first term on the right of \eqref{4.2} is $o(1/t)$. Since
$\abs{H_{\mu_3}(x)} \leq 2/\veps$, the third term is $o(1/t)$. By \eqref{1.4},
the second term is bounded by $C\mu ((x_0-\veps, x_0+\veps)\setminus
\{x_0\})/t$.

So long as $t > \f{2\mu(\{x_0\})}{3\pi\veps}$, the left of \eqref{4.2} is
$\abs{\fre\cap (x_0 - \f{2\mu(\{x_0\})}{3\pi t}, x_0 + \f{2\mu(\{x_0\})}{3\pi
t})}$. Thus, if
\begin{equation} \lb{4.3}
C(x_0) = \limsup_{s\downarrow 0}\, (2s)^{-1}\abs{\fre\cap (x_0-s, x_0+s)}
\end{equation}
\eqref{4.2} implies that
\begin{equation} \lb{4.4}
\f{4C(x_0)\mu(\{x_0\})}{3\pi} \leq C\mu ((x_0-\veps, x_0+\veps)\setminus\{x_0\})
\end{equation}
for any $\veps$. Since $\cap [(x_0-\f{1}{m}, x_0+\f{1}{m})\setminus\{x_0\}]=\emptyset$, the right
side of \eqref{4.4} goes to zero as $\veps\downarrow 0$, and we conclude that $\mu(\{x_0\})=0$.
\end{proof}

To prove Theorem~\ref{T1.6}, we need to describe some sets connected with the
Cantor set. Let $K_1$ be the two connected closed sets $K_{1,1},K_{1,2}$
obtained from $[0,1]$ by removing the middle third. At level $n$, we have $2^n$
intervals $\{K_{n,j}\}_{j=1}^{2^n}$, each with $\abs{K_{n,j}}=3^{-n}$ so
$\abs{K_n}=(\f{2}{3})^n$. The Cantor set, of course, is $K_\infty =\cap K_n$.
The Cantor measure is determined by
\begin{equation} \lb{4.5}
\mu(K_{n,j})=\f{1}{2^n}
\end{equation}

We order $\calI=\{(n,j)\mid n=1,2,\dots,\, j=1,2,\dots, 2^n\}$ with lexigraphic order and use $(n,j+1)$
for the obvious pair if $j<2^n$ and to be $(n+1,1)$ if $j=2^n$. Similarly, $(n,j-1)$ is $(n-1, 2^{n-1})$
if $j=1$.

Let $E_1$ be the middle closed third of $[0,1]\setminus K_1$, so $\abs{E_1}=1/9$. Let $E_2$ be the two
middle thirds of the two gaps in $K_1\setminus K_2$. $E_m$ has $2^{m-1}$ closed intervals of size
$1/3^{m+1}$. There is a unique affine order preserving map of $[0,1]$ to $K_{n,j}$.  Let $E_{n,j,m}$
be the image of $E_m$ under this map, so $E_{n,j,m}$ has $2^{m-1}$ intervals, each of size $1/3^{n+m+1}$,
that is,
\begin{equation} \lb{4.5a}
\abs{E_{n,j,m}}=2^{m-1}/3^{n+m+1}
\end{equation}

We want to pick a positive integer $m(n,j)$ for each $(n,j)\in\calI$ so that
\begin{equation} \lb{4.6}
m(n,j+1) >m(n,j)
\end{equation}
and we define
\begin{equation} \lb{4.7}
k(n,j)=n+m(n,j)
\end{equation}

Given such a choice, we define
\begin{equation} \lb{4.8}
\fre=K_\infty \cup\bigcup_{n,j\in\calI} E_{n,j,m(n,j)}
\end{equation}
Our goal will be to prove $\fre$ is always weakly homogeneous, and that if $m(n,j)$ grows fast enough, then
$H_\mu\restriction\fre$ is in $L_{w;0}^1$.

\begin{lemma}\lb{L4.1} For any choice of $m(n,j)$, $\fre$ is weakly homogeneous. Indeed, for any $x_0\in\fre$,
\begin{equation} \lb{4.9}
\limsup_{\delta\downarrow 0} (2\delta)^{-1} \abs{\fre\cap (x_0-\delta, x_0+\delta)} \geq \f{1}{10}
\end{equation}
\end{lemma}

\begin{proof} Let
\begin{equation} \lb{4.10}
\ti E_{n,j}=E_{n,j,m(n,j)}
\end{equation}
If $x_0\in\ti E_{n,j}$, which is a closed interval, for all small $\delta$,
$(2\delta)^{-1} \abs{\ti E_{n,j} \cap (x_0-\delta, x_0+\delta)} =\f12$ or $1$,
depending on whether $x_0$ is a boundary or an interior point. So \eqref{4.9}
is certainly true.

Thus, we need only consider $x_0\in K_\infty$. Fix $x_0\in K_\infty$. For each
$n$, $x_0\in K_n$, and so in $K_{n,j_n}$ for some $j_n$. Let $k_n\equiv
k(n,j_n)$. On level $k_n$, $x_0$ is contained in some interval, $K_{k_n,\ell}$
of size $3^{-k_n}$ and on one side or the other, there is an interval of size
$3^{-k_n-1}$ in $\ti E_{n,j_n}$ in a touching gap. Let
\begin{equation} \lb{4.11}
\delta_n = \f{5}{3}\, 3^{-k_n}
\end{equation}
Then $(x_0-\delta_n, x_0+\delta_n)$ contains this interval in $\ti E_{n,j_n}$.
Thus,
\begin{equation} \lb{4.12}
(2\delta_n)^{-1} \abs{\fre\cap (x_0-\delta_n, x_0 + \delta_n)} \geq
\f{3^{-k_n-1}}{2\delta_n} = \f{1}{10}
\end{equation}
Since $\delta_n\to 0$ as $n\to\infty$, \eqref{4.9} holds.
\end{proof}

For each $(n,j)$, we will want to define
\begin{equation} \lb{4.13}
\mu_{n,j} = \mu\restriction K_{n,j}\cup K_{n,j-1} \qquad
\ti\mu_{n,j}=\mu - \mu_{n,j}
\end{equation}
that is, single out the part of the Cantor measure near $K_{n,j}$, and so near $E_{n,j}$. We define
\begin{equation} \lb{4.14}
F_{n,j} = F_{\mu_{n,j}} \qquad \ti F_{n,j} = F_{\ti\mu_{n,j}}
\end{equation}

\begin{lemma} \lb{L4.2}
On $\cup_{(\ti n,\ti j)\leq (n,j)} \ti E_{\ti n,\ti j}$, we have
\begin{equation} \lb{4.15}
\abs{\ti F_{n,j}} \leq 3^{k(n,j-1)}
\end{equation}
\end{lemma}

\begin{proof} Since $\norm{\ti\mu_{n,j}}\leq 1$, we have
\begin{equation} \lb{4.16}
\abs{\ti F_{n,j}(x)} \leq \dist (x,K_\infty\setminus K_{n,j-1}\cup K_{n,j}))^{-1}
\end{equation}

By construction,
\begin{equation} \lb{4.17}
\dist(\ti E_{\ti n, \ti j}, K_\infty) = 3^{-k(\ti n, \ti j)-1}
\end{equation}
so if $(\ti n,\ti j)< (n, j-1)$, then for $x\in\ti E_{\ti n, \ti j}$,
\begin{equation} \lb{4.18}
\abs{\ti F_{n,j}(x)} \leq 3^{k(\ti n, \ti j)+1} \leq 3^{k(n,j-1)}
\end{equation}
since $m(\ti n,\ti j)+1 < m(n,j-1)$ implies $k(\ti n, \ti j)+1\leq k(n,j-1)$.

On the other hand, since we have removed $K_{n,j-1}\cup K_{n,j}$,
\begin{equation} \lb{4.19}
\dist(\ti E_{n,j}\cup \ti E_{n,j-1}, K_\infty\setminus (K_{n,j}\cup K_{n,j-1}))
\geq 3^{-n}
\end{equation}
Thus, for $x$ in $\ti E_{n,j}\cup \ti E_{n,j-1}$, we have that
\begin{equation} \lb{4.20}
\abs{\ti F_{n,j}(x)} \leq 3^n \leq 3^{k(n,j-1)}
\end{equation}
proving \eqref{4.15} on the claimed set.
\end{proof}

\begin{proof}[Proof of Theorem~\ref{T1.6}] We construct $\fre$ by using the above construction where
$m(n,j)$ is picked inductively so that
\begin{equation} \lb{4.21}
k(n,j+1) = 3k(n,j)
\end{equation}
By Lemma~\ref{L4.1}, $\fre$ is weakly homogeneous.

Let
\begin{equation} \lb{4.21a}
3^{k(n,j-1)} < t \leq 3^{k(n,j)}
\end{equation}
Since $F_\mu = F_{n,j}+\ti F_{n,j}$, by \eqref{1.7},
\begin{equation} \lb{4.22}
\begin{split}
2t\abs{\{x\in\fre\mid \abs{F_\mu(x)} \geq 2t\}} &\leq 2t \abs{\{x\mid\abs{F_{n,j}(x)} \geq t\}}\\
& \qquad + 2t \abs{\{x\in\fre\mid \abs{\ti F_{n,j}(x)} \geq t\}}
\end{split}
\end{equation}

By Boole's equality \eqref{1.5}, the first term on the right side of
\eqref{4.22} is bounded by
\begin{equation} \lb{4.23}
4(\mu_{n,j-1}(\bbR)+\mu_{n,j}(\bbR)) \leq 4[2^{-n} + 2\cdot2^{-n}] = 12\cdot
2^{-n}
\end{equation}
(where we need the $2\cdot 2^{-n}$ if $j=1$).

By Lemma~\ref{L4.2}, the second term is bounded by
\begin{equation} \lb{4.24}
2\cdot3^{k(n,j)} \sum_{(\ti n, \ti j) \geq (n,j+1)}\, \abs{E_{\ti n, \ti j}}
\end{equation}
By \eqref{4.5a},
\begin{equation} \lb{4.25}
\abs{\ti E_{n,j}} = \f{1}{2^{n+1}\, 3}\, \biggl( \f{2}{3}\biggr)^{k(n,j)}
\end{equation}
so using
\begin{gather}
\sum_{\ell=\ell_0}^\infty \, \biggl( \f{2}{3}\biggr)^\ell = 3\biggl( \f{2}{3}\biggr)^{\ell_0} \lb{4.26} \\
\text{\eqref{4.24}} \leq 3^{k(n,j)}\, 2^{-n} \biggl( \f{2}{3}\biggr)^{k(n,j+1)} \lb{4.26a}
\end{gather}

By \eqref{4.21} and $(\f32)^3 = \f{27}{8} >3$, we see
\begin{equation} \lb{4.27}
\text{\eqref{4.24}} \leq 2^{-n}
\end{equation}
Thus, if $t$ obeys \eqref{4.21a}, then by \eqref{4.22}, \eqref{4.23}, and
\eqref{4.27},
\begin{equation} \lb{4.28}
2t\abs{\{x\in\fre \mid \abs{F_\mu(x)} \geq t\}} \leq 13\cdot2^{-n}
\end{equation}
Since $n\to\infty$ as $t\to\infty$, we see $F_\mu\restriction\fre\in L_{w;0}^1$.
\end{proof}

\section{Non-uniformly Homogeneous Sets} \lb{s5}

Our goal in this section is to prove Theorem~\ref{T1.8} and then also Theorem~\ref{T1.7}. For any Borel set
$\fre$, define
\begin{equation} \lb{5.1}
\fre_n = \biggl\{x\in\fre\biggm| \forall\, a < \f{1}{n}\,,\, \abs{(x-a,x+a)\cap\fre} \geq \f{2a}{n}\biggr\}
\end{equation}

\begin{proposition} \lb{P5.1}
Let $\mu$ be a measure with $\mu(\bbR\setminus\fre_n)=0$. Suppose
$H_\mu\restriction\fre\in L_{w;0}^1$. Then $\mu_\s=0$.
\end{proposition}

\begin{proof} We begin by noting that $\fre_n$ is closed, for if $x_m\to x$ and $\abs{(x_m-a,x_m+a)\cap\fre}
\geq\f{2a}{n}$, then for all $m$,
\begin{equation} \lb{5.2}
\abs{(x-a,x+a)\cap\fre} \geq \f{2a}{n} - 2\abs{x-x_m}
\end{equation}
so $x\in\fre_n$. Applying Theorem~\ref{T1.1} to $d\mu$ and compact homogeneous
sets $\fre_n\cap[-N,N]$ for all $N\geq 1$, we get the result.
\end{proof}

Because $\fre$ is not closed, we cannot use Propositions~\ref{P2.2} and \ref{P2.3} to restrict to $\fre_m$.
Instead we need:

\begin{proposition}\lb{P5.2} Let $\mu$ and $\nu$ be two measures on $\bbR$ whose singular parts are mutually
singular. Then for all $c>0$,
\begin{equation} \lb{5.3}
t\abs{\{x\mid\abs{H_\mu(x)} \geq t\} \cap \{x\mid \abs{H_\nu(x)} \geq ct\}} \to 0
\end{equation}
as $t\to\infty$.
\end{proposition}

\begin{remark}  This result is essentially in Poltoratski \cite{Pol} (see the last set out formula in the proof
of Theorem~2 in that paper), so we only sketch the proof.
\end{remark}

\begin{proof}[Sketch] Suppose first that $c=1$. We begin with what is essentially Theorem~1 of \cite{Pol},
that for any positive measure $\mu$, as $t\to\infty$,
\begin{equation} \lb{5.4}
\tfrac12\, \pi t \chi_{\{x\mid \abs{H_\mu(x)}\geq t\}}\, dx \overset{w}{\longrightarrow} d\mu_\s
\end{equation}
in the weak-$*$ topology. By \eqref{1.6} and \eqref{1.7}, it suffices to prove this for $\mu=\mu_\s$. In that case,
if $\mu^{(\alpha)}$ is the measure with Stieltjes transform,
\begin{equation} \lb{5.5}
F_\alpha(z) = \f{F(z)}{1+\alpha F(z)}
\end{equation}
then (\cite{S250,Pol})
\[
\int_{-(\pi t)^{-1}}^{(\pi t)^{-1}} (d\mu_\alpha(x))\, d\alpha =
\chi_{\{x\mid\abs{H_\mu(x)}\geq t\}} \, dx
\]
so \eqref{5.4} follows from $d\mu_\alpha \overset{w}{\to} d\mu$ as
$\abs{\alpha}\to 0$.

By \eqref{1.8}, if $\mu^{(t)}$ is the measure on the left side of \eqref{5.4},  then
\begin{equation} \lb{5.6}
\norm{\mu^{(t)}}\to \norm{\mu_\s}
\end{equation}

By \eqref{5.4},
\begin{equation} \lb{5.7}
\mu^{(t)} -\nu^{(t)} \overset{w}{\longrightarrow} \mu_\s - \nu_\s
\end{equation}
so
\begin{equation} \lb{5.8}
\liminf\, \norm{\mu^{(t)} - \nu^{(t)}} \geq \norm{\mu_\s -\nu_\s} = \norm{\mu_s} + \norm{\nu_s}
\end{equation}
by the assumed mutual singularity.

But
\begin{equation} \lb{5.9}
\norm{\mu^{(t)} -\nu^{(t)}} = \norm{\mu^{(t)}} + \norm{\nu^{(t)}} - \pi(\text{lhs of \eqref{5.3}})
\end{equation}
\eqref{5.6} and \eqref{5.8} then imply \eqref{5.3} for $c=1$.

This implies the result for $c\geq 1$ and then, by symmetry, for all $c>0$.
\end{proof}

\begin{proof}[Proof of Theorem~\ref{T1.8}] For each $n$, define
\begin{equation} \lb{5.10}
\mu_n = \mu\restriction\fre_n \qquad \nu_n = \mu-\mu_n
\end{equation}
By \eqref{1.7},
\begin{equation} \lb{5.11}
\begin{split}
\abs{\{x\in\fre \mid\abs{H_{\mu_n}(x)} \geq 2t\}} & \leq \abs{\{x\in\fre\mid\abs{H_\mu(x)} \geq t\}} \\
& \quad + \abs{\{x\mid \abs{H_{\mu_n}(x)} \geq 2t,\,
\abs{H_{\nu_n}(x)}\geq t\}}
\end{split}
\end{equation}

By the hypothesis, the first term on the right is $o(1/t)$ and, by
Proposition~\ref{P5.2}, the second is $o(1/t)$. Thus,
$H_{\mu_n}\restriction\fre\in L_{w;0}^1$, and it follows from
Proposition~\ref{P5.1} that $(\mu_n)_\s =0$, that is, $\mu_\s(\fre_n)=0$.

Since
\[
\bigcup_n \fre_n = \bigl\{x\in\fre \bigm| \liminf_{a\downarrow 0}\, (2a)^{-1} \abs{\fre\cap (x-a,x+a)} >0\bigr\}
\]
we have \eqref{1.22x}.
\end{proof}

\bigskip

\end{document}